\documentclass{IEEEtran}
\usepackage[T1]{fontenc}
\usepackage{times}  
\usepackage{helvet}  
\usepackage{courier}  

\usepackage[caption=false,font=footnotesize,labelfont=rm,textfont=rm]{subfig}

\usepackage{enumitem}

\usepackage{newfloat}
\usepackage{listings}
\usepackage{array}
\usepackage{epstopdf}
\usepackage{booktabs}
\newcommand{\descr}[1]{\noindent \textbf{#1}}
\usepackage{multirow}
\usepackage{amsmath}
\usepackage{amsthm}
\usepackage{threeparttable,diagbox}
\usepackage{graphicx}
\usepackage{amssymb}

\theoremstyle{plain}
\newtheorem{thm}{\protect\theoremname}
\theoremstyle{plain}
\newtheorem{lem}[thm]{\protect\lemmaname}
\theoremstyle{definition}
\newtheorem{defn}[thm]{\protect\definitionname}
\theoremstyle{plain}
\newtheorem{cor}[thm]{\protect\corollaryname}
\theoremstyle{plain}

\providecommand{\corollaryname}{Corollary}
\providecommand{\definitionname}{Definition}
\providecommand{\lemmaname}{Lemma}
\providecommand{\propositionname}{Proposition}
\providecommand{\theoremname}{Theorem}

\usepackage{bibentry}

\begin{document}

\title{Algorithmics and Complexity of Cost-Driven Task Offloading with Submodular Optimization in Edge-Cloud Environments}

 \author{
          Longkun~Guo,
          Jiawei~Lin, 
         Xuanming~Xu
         and~Peng~Li
  \IEEEcompsocitemizethanks{
\IEEEcompsocthanksitem
    L. Guo, J. Lin and X. Xu are with School of Mathematics and Statistics, Fuzhou University, Fuzhou, Fujian, 350108, P.R. China.\\
   E-mail: longkun.guo@gmail.com
    \IEEEcompsocthanksitem P. Li is with Google LLC, Seattle, WA, United States.
}}


\maketitle

\begin{abstract}
Emerging applications such as autonomous driving pose the challenge of efficient cost-driven offloading in edge-cloud environments. This involves assigning tasks to edge and cloud servers for separate execution, with the goal of minimizing the total service cost including communication and computation costs. In this paper, observing that the intra-cloud communication costs are relatively low and can often be neglected in many real-world applications, we consequently introduce the so-called communication assumption which posits that the intra-cloud communication costs are not higher than the inter-partition communication cost between cloud and edge servers, nor the cost among edge servers. As a preliminary analysis,  we first prove that the offloading problem without the communication assumption is NP-hard, using a reduction from MAX-CUT. Then, we show that the offloading problem can be modeled as a submodular minimization problem, making it polynomially solvable.  Moreover, this polynomial solvability remains even when additional constraints are imposed, such as when certain tasks must be executed on edge servers due to latency constraints.  By combining both algorithmics and computational complexity results, we demonstrate that the difficulty of the offloading problem largely depends on whether the communication assumption is satisfied.  Lastly,  extensive experiments are conducted to evaluate the practical performance of the proposed algorithm, demonstrating its significant advantages over the state-of-the-art methods in terms of efficiency and cost-effectiveness.
\end{abstract} 





\section{Introduction}




 The rapid advancement of AI technologies imposes significant challenges for resource allocation in edge-cloud environments, necessitating the optimal integration of both edge and cloud resources. As smart devices ranging from smartphones and computers to internet-connected appliances are generating vast volumes of data,   more sophisticated resource allocation strategies are increasingly needed to manage this data explosion effectively.
 Traditional cloud infrastructures, which began showing their limitations decades ago, have struggled to meet the demands of the expanding Internet of Things (IoT) and intelligent mobile devices, primarily due to bandwidth constraints and resource limitations \cite{shi2016edge}. In response, edge computing has emerged as a promising solution, allowing portions of computational tasks to be executed directly on edge servers, thereby reducing dependence on distant cloud centers and alleviating latency issues \cite{xiao2019edge}.
Edge computing is particularly advantageous in scenarios requiring real-time processing and decision-making, such as autonomous driving. Autonomous vehicles generate substantial real-time data that must be processed instantly to ensure timely and safe navigation decisions \cite{liu2019edge}. Processing data locally on edge servers in such cases significantly reduces latency compared to relying solely on centralized cloud resources, thereby mitigating the risks associated with delayed responses \cite{DingPotential2022}. Moreover, edge computing enhances privacy by enabling sensitive data --- such as that generated by wearable devices, industrial machinery, and other equipment ---  to be processed closer to the data source, reducing the need for transmission over potentially insecure networks \cite{shi2016edge}.

Among the emerging technologies edge-cloud environment employs to effectively bridge the gap between centralized cloud infrastructure and end-users where near-data computing services are with minimum latency\cite{YueTODG2022},  offloading stands out as one of the most promising strategies to enhance service quality, energy efficiency, and resource utilization in edge-cloud environments. By leveraging the complementary advantages of both edge and cloud computing, offloading ensures dynamic, reliable, and quality-of-service (QoS)-optimized IT services \cite{wang2020survey}.
As illustrated in Fig. \ref{fig:Edge-cloud}, a typical architecture of an edge-cloud environment consists of three functional layers: cloud servers, edge servers, and edge devices. Cloud servers possess abundant computational power and resources but suffer from high latency and potential privacy concerns.  In contrast, edge servers can provide timely processing, albeit with more limited computational resources. The offloading problem, therefore, aims to execute tasks collaboratively across cloud servers and edge devices to optimize both computation costs and delays.

In this paper, we focus on the cloud offloading problem, as studied in \cite{MahmoodiCloudOffloading2019,HuamingCloudOffloading2020}, which involves assigning tasks to either edge servers or cloud servers for separate execution. The goal is to minimize the total service cost in the edge-cloud environment, including both computation and communication costs.


\begin{figure*}[t]
    \centering
    \includegraphics[scale=0.4]{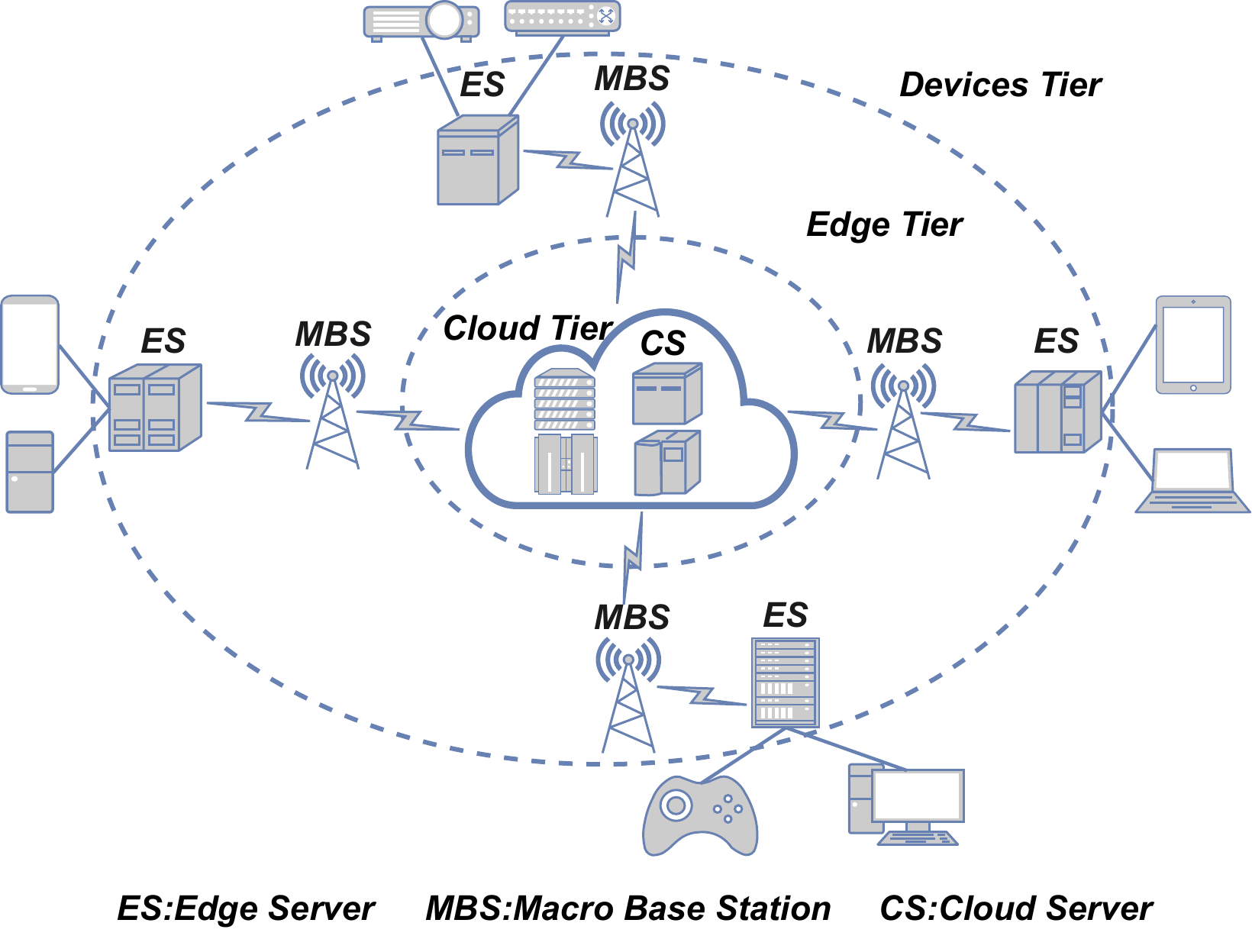}
    \caption{ A typical configuration of the edge-cloud environment.}
    \label{fig:Edge-cloud}
\end{figure*}

\subsection{Related Work}

The offloading problem has been intensively studied in the context of mobile cloud computing \cite{kumar2010cloud,kumar2013survey}, for which many offloading technologies were developed regarding the edge-cloud environment \cite{wang2020survey}. For the offloading problem that minimizes the total cost consisting of communication and computation costs, the input can be modeled as a weighted directed graph, where the vertices of the graph correspond to the tasks (or databases) to be offloaded, and the directed edge between two vertices represent the data communication between two tasks (or two databases). The weight of the vertex and edge denotes the computation cost and the communication cost, respectively. The goal is to divide the graph into two disjoint sub-graphs, one of which is executed at the edge server while the other is at the cloud server. In the traditional model, previous researches assume the communication costs between two tasks on the same side can be ignored, while the costs between two tasks on different sides are symmetric. Based on such assumptions, many algorithms were designed to obtain the minimum service cost in offloading \cite{wu2016optimal,wang2017computational,wang2018cost,dong2018computation}.
However, those assumptions do not represent real-world applications.
Hence, Du \textit{et al.} \cite{du2020algorithmics} proposed a more general model called the heterogeneous model, in which they considered: (1) The communication costs between any two tasks on the same side can not be ignored; (2) The communication costs between edge and cloud servers are asymmetric. Then they showed that the offloading problem with the heterogeneous model is NP-hard and designed an efficient greedy heuristic algorithm. Moreover, they transferred offloading to a classical Min-Cut problem when the model is homogeneous. Later, Han \textit{et al.}\textit{ }\cite{han2020approximation} studied the offloading problem in the heterogeneous model with cardinality constraint and proposed a $\frac{2}{\pi}$-approximation algorithm based on semidefinite programming.

Offloading with other optimization objectives was also studied regarding
the edge-cloud environment. 
Hao \textit{et al. }\cite{hao2018energy} investigated the problem of joint optimization of task caching and offloading with computing and storage resource constraints for the sake of energy efficiency of mobile devices and meeting real-time requirements. They formulated the problem via mixed-integer programming and designed an efficient algorithm. 
To minimize the offloading cost while maintaining performance guarantees, Chen \textit{et al.}\cite{chen2018dynamic} decomposed computation offloading into sub-problems and proposed a dynamic algorithm based on stochastic optimization. 
Observing load balancing as a great concern in the edge-cloud environment, Lim \textit{et al.} \cite{lim2020load} proposed an algorithm by a graph coloring-based implementation using a genetic algorithm to distribute the offloaded tasks to nearby edge servers. 
To assign tasks to the edge servers and minimize the total response time for completing all the tasks, Fang \textit{et al.}\cite{fang2019job} and Tan \textit{et al. }\cite{tan2017online} provided an online approximate 
algorithm.
Luo \textit{et al.} \cite{Luo2024computation} focused on optimizing device computation offloading and proposed a DQN-based resource utilization task scheduling algorithm to minimize overall task processing latency in the system.
To meet the maximum number of deadlines, Meng \textit{et al.}\textit{
}\cite{meng2019dedas} proposed an online greedy algorithm by jointly
considering the management of network bandwidth and computing resources.
Ji \textit{et al.} \cite{ji2023towards} transformed the application offloading problem in heterogeneous edge cloud environments into a minimum cut optimization problem to minimize total latency under given boundary conditions. 
Intending to jointly optimize the latency and energy consumption, Wang \textit{et al.} \cite{Wang2022Dependent} proposed an intelligent task offloading scheme leveraging off-policy reinforcement learning empowered by a Sequence-to-Sequence (S2S) neural network.
Kai \textit{et al.} \cite{Kai2021collaborative} proposed a pipeline-based offloading scheme, formulating the problem of minimizing the sum latency of all mobile devices as a non-convex problem. They utilized the successive convex approximation method to convert it into a convex optimization problem for resolution.

Besides offloading between the edge and cloud servers, some researchers
focus on optimizing the layout of edge servers for better utilization
of edge resources and quality of service. To maximize rewards, Pasteris
\textit{et al.}\cite{pasteris2019service} designed an efficient
offline algorithm with a constant approximate ratio for the problem
that is known as NP-hard. To improve the quality of service, Gao
\textit{et al. }\cite{gao2019winning} formulated offline network
selection and service layout as an optimization problem and proposed
a competitive algorithm based on three delay assumptions. Furthermore,
to timely and efficiently schedule services concerning their requests,
Farhadi \textit{et al.} \cite{farhadi2019service} proposed a system
model that allows services to be migrated between different edge servers.
To improve the performance of personalized service, Ouyang \textit{et
al.} \cite{ouyang2019adaptive} designed a new adaptive user management
service placement mechanism that can tune services according to individual
needs.

In the paper, we focus on the submodular function that has a natural
diminishing returns property and arises in various fields of operations
research such as approximation algorithms, discrete optimization,
game theory, and image collection summarization \cite{schrijver2003combinatorial}.
We transform the offloading problem with communication assumption (to be given later) to submodular minimization, for which Gr\"{o}tschel \textit{et al.
}were the first to design a polynomial algorithm \cite{grotschel1981ellipsoid}.
Later, Cunningham \textit{et al.} \cite{cunningham1985submodular}
proposed a combinatorial strongly polynomial algorithm to compute
the minimum value of a general submodular function. Meanwhile, Iwata
\textit{et al.} \cite{iwata2001combinatorial} independently developed
another combinatorial strongly polynomial algorithm. In contrast,
submodular maximization is known as NP-hard \cite{lewis1983computers}
but admits $\frac{1}{2}$-approximation when the function is non-negative
\cite{buchbinder2015tight}. Differently, it is NP-hard to solve submodular
minimization with a simple constraint like cardinality constraint
\cite{svitkina2011submodular}, and submodular minimization with a
cardinality constraint admits no constant-factor approximation. Interestingly,
there exists $\left(1-\frac{1}{e}\right)$-approximation algorithm
for monotone submodular maximization with a cardinality, knapsack
or matroid constraint \cite{sviridenko2004note}.

\subsection{Our Contributions}

In this paper, we address the offloading problem in edge-cloud environments,
aiming to minimize the total service cost that consists of both communication and computation costs. Our contributions can be summarized as follows:

\begin{itemize}
    \item Prove that the symmetric version of the offloading problem is NP-hard by reducing from MAX-CUT \cite{garey1974some}, thereby resolving the open problem of whether the symmetric offloading problem is NP-hard. This result complements and generalizes the previous NP-hardness proof for the asymmetric version by a reduction from 3SAT  \cite{du2020algorithmics}.

\item Propose an optimal algorithm for
  the offloading problem based on submodular minimization under the assumption that the inter-partition communication cost (i.e., the communication cost between edge and cloud servers) is not less than the intra-cloud communication cost (i.e., the communication cost between cloud nodes).

\item Conduct extensive experiments to evaluate the practical performance of our algorithm, demonstrating that it significantly outperforms the state of the art. 
\end{itemize}

 To our knowledge, our algorithm is the first to use submodular minimization to optimally solve the offloading problem while achieving a polynomial runtime. Moreover, based on our hardness and algorithmic analysis, the offloading problem is NP-hard when the communication assumption is absent, but becomes polynomially solvable when the assumption holds.  This result 
 reveals that the communication assumption is the key factor in determining whether the offloading problem is NP-hard, serving as the dividing line between $NP$-hardness and polynomial solvability.

\section{Preliminaries \label{sec:Preliminaries}}
In this section, we present mathematical notations and symbols
to capture the characteristics of the offloading problem, where some of them follow a similar line as in \cite{du2020algorithmics,han2020approximation}.
We describe the offloading problem with latency constraints as
a weighted directed graph $G=\left(V,\,T,\,E,\,w,\,l\right)$.
\begin{itemize}
    \item  $V$ represents the node set, each node $v_{i}\in V$ 
corresponding to a computational task (or database) ready to be offloaded, which
 can be executed at either the edge or the cloud server. We assume
that there are $n$ tasks in the offloading, denoted by $V=\left\{ v_{1},\,v_{2},\,\cdots,\,v_{n}\right\} $.

 $T\subseteq V$ represents the set of real-time tasks with latency constraints, which can only be processed on edge servers. When $T\neq\emptyset$,
 the problem is called \textit{latency-constrained}.

\item $E$ is the directed edge set that represents the set of communication requirements
between pairs of nodes. Each directed edge $e\left(v_{i},\,v_{j}\right)$
that leaves $v_{i}$ and enters $v_{j}$ indicates there exist communication
requirements from task $v_{i}$ to $v_{j}$. Assume that there are
$m$ directed edges in the offloading problem, i.e., $\left|E\right|=m$.

\item Node weight $w_{i}$ indicates the computation cost for executing
task $v_{i}$. In real-world applications, task $v_{i}$ consumes
a computation cost at the cloud server different from that at the
edge server, because the configuration and performance of the two
sides are different. Thus, node weight $w_{i}$ can be denoted as
a 2-tuple:
\begin{equation}
w_{i}=\left(w_{i}^{edg},\,w_{i}^{cld}\right),\label{eq:node weight}
\end{equation}
where $w_{i}^{edg}$ and $w_{i}^{cld}$ indicate the computation cost
of task $v_{i}$ executed at the edge and the cloud server, respectively.
In addition, there may exist some static tasks that can not be offloaded
to the other side. For example, if task $v_{i}$ is with $w_{i}=\left(w_{i}^{edg},\,+\infty\right)$,
then it must be computed at the edge server; while if $v_{i}$ is
with $w_{i}=\left(+\infty,\,w_{i}^{cld}\right)$, it must be processed
at the cloud server.

\item Each weight $l_{ij}$ indicates the communication cost from task $v_{i}$
to $v_{j}$. Note that there are four possible values for the communication
cost $l_{ij}$, depending on whether $i$ and $j$ are processed on
the clouds or the edge server. More precisely, we associate each edge
with four weights as follows:
\begin{equation}
l_{ij}=\left(l_{ij}^{1},\,l_{ij}^{2},\,l_{ij}^{3},\,l_{ij}^{4}\right),\label{eq:edge weight}
\end{equation}
where each element in the 4-tuple has its own value: $l_{ij}^{1}=l\left(v_{i}^{edg},\,v_{j}^{edg}\right)$,
$l_{ij}^{2}=l\left(v_{i}^{edg},\,v_{j}^{cld}\right)$, $l_{ij}^{3}=l\left(v_{i}^{cld},\,v_{j}^{edg}\right)$
and $l_{ij}^{4}=l\left(v_{i}^{cld},\,v_{j}^{cld}\right)$. The element
$v_{i}^{x},\,x\in\left\{ edg\left(Edge\right),\,cld\left(Cloud\right)\right\} $
means the task $v_{i}$ is executed at $x$ side (the edge or the
cloud server), and hence $l\left(v_{i}^{x},\,v_{j}^{y}\right)$ indicates
the communication cost when task $v_{i}$ located at $x$ side and
$v_{j}$ located at $y$ side. Then we say the offloading problem is \emph{symmetric} if $l_{ij}^{1} = l_{ij}^{4}$ and $ l_{ij}^{2}= l_{ij}^{3}$ both hold, and is \emph{asymmetric} otherwise.

\end{itemize}

For briefness, we call $l_{ij}^{1}$ and $l_{ij}^{4}$ as intra-edge and intra-cloud costs, respectively, and  call $l_{ij}^{1}$ and $l_{ij}^{4}$ as inter-partition cost.  


\begin{defn}[communication
assumption] We say the communication assumption is true for an instance of the offloading problem, if and only if $l_{ij}^{4}\le l_{ij}^{2},\,l_{ij}^{3}$. In other words, the assumption means, the intra-cloud communication cost is not larger than the inter-partition cost.
    
\end{defn}
Note that, the communication assumption is
natural since the intra-cloud communication costs are not higher
than the inter-partition communication costs in many real-world applications.
The phenomenon is because there are significantly better connections between the servers in clouds in comparison to the network connections between clouds and edge servers \cite{prakash2018storage}.

The offloading problem aims to find a partition for $V$, say $V_{edg},\,V_{cld}$
with $V_{cld}\cup V_{edg}=V$ and $V_{cld}\cap V_{edg}=\emptyset$,
such that the total service cost is minimized. Intuitively, $V_{cld}\subseteq V$
and $V_{edg}\subseteq V$ represent the sets of tasks to be processed
at the cloud and the edge server, respectively. In the following,
we will further formally define the total cost of the offloading
problem. The total weight of nodes based on the partition $V_{edg},\,V_{cld}$
as the total computation cost can be calculated by the following formulation:
\begin{align}
C_{comp}\left(V_{G}\right)= & C_{comp}\left(V_{edg},\,V_{cld}\right)\nonumber \\
= & \underset{i\in I(V_{edg})}{\sum}w_{i}^{edg}+\underset{i\in I(V_{cld})}{\sum}w_{i}^{cld}\text{,}\label{eq:computation cost}
\end{align}
where $I\left(X\right)$ denotes the index set of tasks in $X$.
{Moreover, note that there is an overhead cost when transferring
a task from the cloud to the edge server (or the other way around, i.e. edge to the cloud)}. However, this cost can be
merged into the computation cost. Denoting the cost of transferring
task $v_{i}$ by $t_{i}$, we can easily combine the transferring
cost with the computation cost by incorporating $t_{i}$ into Equality
(\ref{eq:computation cost}) and obtain:
\begin{align}
C_{comp}\left(V_{G}\right)= & \underset{i\in I(V_{edg})}{\sum}\left(w_{i}^{edg}+t_{i}\right)+\underset{i\in I(V_{cld})}{\sum}w_{i}^{cld}.\label{eq:trans+comp}
\end{align}

Likewise, the total communication cost based on the partition $P$
can be defined as follows:
\begin{align}
C_{comm}\left(E_{G}\right)= & C_{comm}\left(V_{edg},\,V_{cld}\right)\nonumber \\
= & C_{comm}\left(V_{edg},\,E_{V_{cld}}\right)+C_{comm}\left(G\left[V_{edg}\right]\right) \\
&+C_{comm}\left(G\left[V_{cld}\right]\right),
\end{align}
where $E({V_{edg},\,V_{cld})}$ (or equivalently $E_{V_{edg},\,V_{cld}}$) denote the edges between $V_{cld}$ and $V_{edg}$, and $G\left[V_{cld}\right]$ is the contraction  of $G$ on vertex set $V_{cld}$ that is the set of edges between vertices of $V_{cld}$ in $G$.

Finally, the offloading problem is to minimize the total cost
as below:
\begin{align}
C_{total}\left(G\right)= & C_{total}\left(V_{edg},\,V_{cld}\right)\nonumber \\
= & C_{comp}\left(V_{edg},\,V_{cld}\right)+C_{comm}\left(V_{edg},\,V_{cld}\right)\nonumber \\
= & \underset{i\in I\left(V_{edg}\right)}{\sum}\left(w_{i}^{edg}+t_{i}\right)+\underset{i\in I\left(V_{cld}\right)}{\sum}w_{i}^{cld}\nonumber \\
 & +C_{comm}\left(E\left({V_{edg},\,V_{cld}}\right)\right)+C_{comm}\left(G\left[V_{edg}\right]\right)\nonumber \\
& +C_{comm}\left(G\left[V_{cld}\right]\right).\label{eq:Ctotal}
\end{align}

Throughout this paper, we use \emph{Offloading-CommA } to denote the offloading problem with the communication assumption and further use 
\emph{Offloading-LCCA} to denote \emph{Offloading-CommA }  the offloading problem with latency constraint. 
In the following paragraph, we will first show the NP-hardness of the offloading problem, and then provide polynomial algorithms for both \emph{Offloading-CommA } and \emph{Offloading-LCCA}.

\section{NP-Hardness of Symmetric Offloading 
\label{sec:hard symmetric}}

In this section, we mainly focus on the hardness of the offloading problem without considering communication assumptions, and actually demonstrate the NP-hardness of the offloading problem --- even when the inter-partition communication costs are symmetric and the computation costs could be ignored (\emph{a.k.a. symmetric offloading problem}) --- by reducing from MAX-CUT that is known NP-hard~\cite{garey1974some}.

In an instance of MAX-CUT, we
are given a positive integer $k$ and an undirected graph $G=\left(N,\,A\right)$,
where the numbers of nodes and edges are respectively $\left|N\right|=n$
and $\left|A\right|=m$. The aim is to determine whether there is a cut that divides the vertices of $N$ into two components, such that
the number of edges in the cut is not less than $k$. 

\begin{thm}
\label{thm:max-cut}The symmetric offloading problem is NP-hard.
\end{thm}
We need only to give a reduction from MAX-CUT 
to the symmetric offloading problem.  The transformation
from MAX-CUT to the decision symmetric offloading problem proceeds
as below:
\begin{itemize}[leftmargin=9pt]
\item Tasks (nodes):
\begin{itemize}
\item For each vertex $u_{i}\in N$, add a task $v_{i}$ into  task set
$V$;
\item Set the cost of $v_{i}$ as $w_{i}^{edg}=w_{i}^{cld}=0$.
\end{itemize}
\item Relationship among tasks (edges):
\begin{itemize}
\item For each edge $e\left(u_{i},\,u_{j}\right)\in A$, add to
edge set $E$ an edge $e\left(v_{i},\,v_{j}\right)$ which corresponds
to the data communication between task $v_{i}$ and $v_{j}$;
\item Set the cost of each edge $e\left(v_{i},\,v_{j}\right)\in E$ as $l_{ij}^{1}=l_{ij}^{4}=m$
and $l_{ij}^{2}=l_{ij}^{3}=1$.
\end{itemize}
\end{itemize}
Then, the aim of the symmetric offloading problem is to determine whether there exists a feasible partition
for $V$ with the total cost not larger than $C=m\left(m-k\right)+k$.

Note that following the computation and communication costs constructed above, the computation costs are all zero while the communication costs are symmetric. Then the correctness of Theorem \ref{thm:max-cut} can be immediately obtained from the following lemma:

\begin{lem}
\label{lem:max-cut lemma} An instance of decision symmetric offloading problem has an offloading solution with the total cost being at most $C=m\left(m-k\right)+k$ if the corresponding instance of MAX-CUT is feasible.
\end{lem}

\begin{proof}
Suppose there exists a feasible offloading solution $\left(V_{edg},\,V_{cld}\right)$,
whose total cost consisting of both computational and communication costs is $C^\prime$. Then, $C'\leq C$ apparently holds because $\left(V_{edg},\,V_{cld}\right)$ is a feasible solution. By choosing
the edges between $N_{1}=V_{edg}$ and $N_{2}=V_{cld}$ as the cut $\left(N_{1},\,N_{2}\right)$, it remains only to show $\left(N_{1},\,N_{2}\right)$ is a feasible solution to MAX-CUT. Firstly and obviously,
$\left(N_{1},\,N_{2}\right)$ is a cut of $G$ because it includes
every edge between $N_{1}$ and $N_{2}$ and hence can separate them.
Secondly, we will show the number of edges in $\left(N_{1},\,N_{2}\right)$
is not less than $k$. Let the number of edges between any two tasks
at different sides be $q$. Then, since $l_{ij}^{1}=l_{ij}^{4}=m$
and $l_{ij}^{2}=l_{ij}^{3}=1$ by the transformation, we can calculate the total cost of the cut as below:
\begin{equation*}
C'=m\left(m-q\right)+q.
\end{equation*}
On the other hand, we have $ C= m\left(m-k\right)+k$ and $C' \le C$. Combining with the above inequality, we  immediately have $q\ge k$, and get that $\left(N_{1},\,N_{2}\right)$, $N_{1}=V_{edg}$ and $N_{2}=V_{cld}$, is a feasible solution to the instance of MAX-CUT.

Conversely, assume that there exists a feasible cut $\left(N_{1},\,N_{2}\right)$
with the number of edges therein being $q\ge k$. Then we partition the set of tasks $V$ into $V_{edg}=N_1$ and $V_{cld}=N_2$, according to the cut $\left(N_{1},\,N_{2}\right)$. Because the computation cost of each task is zero by the transformation, the total cost of this solution is equal to $m\left(m-q\right)+q$, which is then bounded by  $C= m\left(m-k\right)+k$.  This completes the proof.
\end{proof}

Following the aforementioned transformation for proving Theorem \ref{thm:max-cut},
we can get a more general computational complexity result:
\begin{cor}
The offloading problem is NP-hard even when all the following conditions
hold: 
\begin{itemize}
    \item [(1)] The computation costs are zero; 
    \item [(2)] The intra-cloud and intra-edge
communication costs are equal to $m$;
    \item[(3)]  The inter-partition communication costs are symmetric and equal to $1$.
\end{itemize} 
\end{cor}

Interestingly, from the NP-hardness proof, we observe that the hardness
of \emph{symmetric offloading problem} comes from the lack of relationship of the communication costs. In fact, the communication assumption essentially decides whether the offloading problem is NP-hard: When the assumption is true, even \emph{offloading-LCCA}, a more general version that is asymmetric and with a real-time constraint can be solved in polynomial time; when without the assumption, even the symmetric version of the problem is NP-hard.

In previous literature, Du et al. \cite{du2020algorithmics} presented  an NP-hard proof for
the NP-hardness of the offloading problem with \textit{asymmetric} communication
costs. It remains open whether
the offloading problem is NP-hard when the inter-partition cost is
symmetric. Thus, our result close the open problem and is more general in comparison.

\section{Latency-Constrained\emph{ Offloading-CommA} \label{sec:Latency-Constrained-Offloading-C}}

In this section, we first introduce an equivalent form of \emph{Offloading-LCCA}, and then arguably show that it can be modeled with a submodular function. Thus, by employing the existing polynomial algorithm for submodular
minimization \cite{iwata2001combinatorial}, we obtain a polynomial-time algorithm that always produces optimal solutions
for \emph{Offloading-LCCA}. Lastly, we illustrate an instance of \emph{Offloading-LCCA}
for which the previous algorithm based on Min-Cut \cite{du2020algorithmics}
fails while our algorithm succeeds in finding the optimum solution, indicating that our method is more general in comparison.

\subsection{Function of Cost Difference }

Observing the hardness of analyzing the original objective function
of \emph{Offloading-LCCA}, we describe an equivalently complementary form of \emph{Offloading-LCCA} by considering the cost \textit{increment} when respectively executing task sets $X\subseteq\left(V\backslash T\right)$ and $V\setminus X$ in the cloud and the edge, 
{comparing with offloading all the tasks $V$ including $X$ 
from the cloud to the edge server}. 
Let $\Gamma\left(Y\right)=C_{total}\left(V\setminus Y, \,Y\right)$ for briefness. That is, $\Gamma\left(Y\right)$ denotes the total cost when executing $Y$ in the cloud and the other part in the edge. Then,
we define the function $F\left(X\right)$
capturing the cost \textit{increment}  formally  in the following:
\begin{alignat}{1}
F\left(X\right) = & \Gamma(X)-\Gamma(\emptyset) \\
= & C_{total}\left(V\setminus X, \,X\right)-C_{total}\left(V,\,\emptyset\right).\label{eq:F(X)}
\end{alignat}
In particular, we have $F\left(\emptyset\right)=0$. Then we 
show finding an optimum $X$ with $F\left(X\right)$ minimized is
equivalent to \emph{Offloading-LCCA,} as stated below:
\begin{lem}
\label{lem:min} Offloading-LCCA can be optimally solved by finding
$X^{*}$ with 
\[F\left(X^{*}\right)=\min_{X\subseteq V\setminus T}\left\{ F\left(X\right)\right\}\leq 0.\]
\end{lem}
\begin{proof}
First of all, following the minimality of $F\left(X^{*}\right)$, we immediately have 
\[F\left(X^{*}\right)\leq F\left(\emptyset\right)=0. \]
Let $SOL$ denote a solution to \emph{Offloading-LCCA}. Recall that
we initially offload all the 
tasks $V$ to the edge server, where the initial cost is denoted
by $\Gamma\left(\emptyset\right)=C_{total}\left(V,\,\emptyset\right)$.
Assume that we offload a set of tasks $X$ to the cloud server afterward,
then 
\[
f\left(SOL\right)=\Gamma\left(\emptyset\right)+F\left(X\right).
\]
Thus, the cost of $\Gamma\left(\emptyset\right)$ is  fixed  as $V$
is given at the beginning. Therefore, we need only to find $X^{*}$
with minimized $F\left(X\right)$ so as to minimize $f\left(SOL\right)$.
This completes the proof.
\end{proof}

\subsection{Submodularity of the function }
Formally, a submodular function can be defined as follows:
\begin{defn}
\cite{schrijver2003combinatorial} \label{def:submodular function}
Let ${\displaystyle \Omega}$ be a finite set and ${\displaystyle 2^{\Omega}}$
denote the power set of ${\displaystyle \Omega}$. Then set function
${\displaystyle F:2^{\Omega}\rightarrow\mathbb{R}}$ is submodular
if and only if it satisfies the following condition: For every ${\displaystyle A,\,B\subseteq\Omega}$
with ${\displaystyle A\subseteq B}$ and every ${\displaystyle x\in\Omega\setminus B}$,
we have that ${\displaystyle F\left(A\cup\left\{ x\right\} \right)-F\left(A\right)}$$\geq F\left(B\cup\left\{ x\right\} \right)-F\left(B\right)$.
\end{defn}

Then we  show the main result of this section that the \textit{cost increment}
function $F\left(X\right)$ as in Equality (\ref{eq:F(X)}) satisfies
the above condition and hence is submodular. Recall that $F\left(X\right)$
represents the cost increment incurred by offloading the task set
$X\subseteq\left(V\backslash T\right)$ from the edge servers to the cloud,
we then immediately conclude that \emph{Offloading-LCCA} can be optimally solved in polynomial runtime by employing submodular minimization.
\begin{thm}
\label{thm:submodular minimization} The cost increment function \textup{$F\left(X\right)$}
of Offloading-LCCA is submodular.
\end{thm}
\begin{proof}
Assume that all tasks of $V$ are initially located 
at the edge server, including the tasks of $T$ that must be executed at the edge
server because of the latency constraint. We consider two sets $A$
and $B$, $A\subseteq B\subseteq\left(V\backslash T\right)$. Let $v_{m}\in\left(V\backslash T\right)\backslash B$ be the task to be offloaded.

Following Definition \ref{def:submodular function}, we
use $F:2^{V}\rightarrow \mathbb{R}$ as the function for capturing
the cost decrement of offloading, where ${\displaystyle 2^{V}}$ denotes
the power set of ${\displaystyle V}$. Then it remains only to show
\[F\left(A\cup\left\{ v_{m}\right\} \right)-F\left(A\right)\ge F\left(B\cup\left\{ v_{m}\right\} \right)-F\left(B\right).\]

For the left side,
we have
\begin{align}
 & F\left(A\cup\left\{ v_{m}\right\} \right)-F\left(A\right)\label{eq:F2=00005Cgamma}\\
= & \left(\Gamma\left(A\cup\left\{ v_{m}\right\} \right)-\Gamma(\emptyset)\right)-\left(\Gamma\left(A\right)-\Gamma(\emptyset)\right).\nonumber\\
= & \Gamma\left(A\cup\left\{ v_{m}\right\} \right)-\Gamma\left(A\right).\nonumber
\end{align}

By the definition of $\left(\Gamma\right)$ and by denoting $T_{1}=A\cup\left\{ v_{m}\right\}$, we have
\begin{alignat}{2}
 & \Gamma\left(T_{1}=A\cup\left\{ v_{m}\right\} \right) \nonumber\\
  =  & C_{comp}\left(V\setminus T_{1},\text{\,}T_{1}\right)
   +C_{comm}\left(V\setminus T_{1},\,T_{1}\right).\label{eq:wholeT1}
\end{alignat}
By the definition of $C_{comp}\left(\cdot\right)$, we have
\begin{equation}
C_{comp}\left(V\setminus T_{1},\text{\,}T_{1}\right)=\underset{i\in I\left(T_{1}\right)}{\sum}w_{i}^{cld}+\underset{i\in I\left(V\setminus T_{1}\right)}{\sum}w_{i}^{edg}.\label{eq:comp}
\end{equation}
Similarly, by the definition of $C_{comm}\left(\cdot\right)$, we
have

\noindent{\begin{eqnarray}
 &  & C_{comm}\left(V\setminus T_{1},\,T_{1}\right) \label{eq:comm}\\
 &  =& C_{comm}\left(E_{V\setminus T_{1},\,T_{1}}\right)+C_{comm}\left(G\left[T_{1}\right]\right) \nonumber \\
 && +C_{comm}\left(G\left[V\setminus T_{1}\right]\right)\nonumber \\
 & = & \underset{V_{cld}=T_{1}=A\cup\left\{ v_{m}\right\} }{\underbrace{\underset{j\in I\text{\ensuremath{\left(V\setminus T_{1}\right)}}}{\sum}\underset{i\in I\left(T_{1}\right)}{\sum}l\left(v_{i},\,v_{j}\right)+\underset{j\in I\left(T_{1}\right)}{\sum}\underset{i\in I\left(T_{1}\right)}{\sum}l\left(v_{i},\,v_{j}\right)}}\nonumber \\
 & &+\underset{V_{cld}=T_{1}=A\cup\left\{ v_{m}\right\} }{\underbrace{\underset{j\in I\left(V\setminus T_{1}\right)}{\sum}\underset{i\in I\left(V\setminus T_{1}\right)}{\sum}l\left(v_{i},\,v_{j}\right)}},\nonumber
\end{eqnarray}}
where $l\left(v_{i},\,v_{j}\right)$ represents the corresponding
communication cost between task $v_{i}$ and $v_{j}$ as stated in
Section \ref{sec:Preliminaries}, and if there is no communication
requirement between task $v_{i}$ and $v_{j}$ or $i=j$, we have
$l\left(v_{i},\,v_{j}\right)=0$.

Combining Equalities (\ref{eq:wholeT1}) (\ref{eq:comp}) and (\ref{eq:comm})
immediately yields:
\begin{align}
 & \Gamma\left(T_{1}=A\cup\left\{ v_{m}\right\} \right)\label{eq:T1}\\
= & \underset{i\in I\left(T_{1}\right)}{\sum}w_{i}^{cld}+\underset{i\in I\left(V\setminus T_{1}\right)}{\sum}w_{i}^{edg}+\underset{V_{cld}=T_{1}=A\cup\left\{ v_{m}\right\} }{\underbrace{\underset{j\in I\text{\ensuremath{\left(V\setminus T_{1}\right)}}}{\sum}\underset{i\in I\left(T_{1}\right)}{\sum}l\left(v_{i},\,v_{j}\right)}} \nonumber\\
 & +\underset{V_{cld}=T_{1}=A\cup\left\{ v_{m}\right\} }{\underbrace{\underset{j\in I\left(T_{1}\right)}{\sum}\underset{i\in I\left(T_{1}\right)}{\sum}l\left(v_{i},\,v_{j}\right)+\underset{j\in I\left(V\setminus T_{1}\right)}{\sum}\underset{i\in I\left(V\setminus T_{1}\right)}{\sum}l\left(v_{i},\,v_{j}\right)}.}\nonumber
\end{align}

Similarly, for $A$, we have
\begin{align}
 & \Gamma\left(T_{2}=A\right)\label{eq:T2}\\
= & \underset{i\in I\left(T_{2}\right)}{\sum}w_{i}^{cld}+\underset{i\in I\left(V\setminus T_{2}\right)}{\sum}w_{i}^{edg}+C_{comm}\left(E_{V\setminus T_{2},\,T_{2}}\right)\nonumber \\
 & +C_{comm}\left(G\left[T_{2}\right]\right)+C_{comm}\left(G\left[V\setminus T_{2}\right]\right)\nonumber \\
= & \underset{i\in I\left(T_{2}\right)}{\sum}w_{i}^{cld}+\underset{i\in I\left(V\setminus T_{2}\right)}{\sum}w_{i}^{edg}+\underset{V_{cld}=T_{2}=A}{\underbrace{\underset{j\in I\text{\ensuremath{\left(V\setminus T_{2}\right)}}}{\sum}\underset{i\in I\left(T_{2}\right)}{\sum}l\left(v_{i},\,v_{j}\right)}}\nonumber \\
 & +\underset{V_{cld}=T_{2}=A}{\underbrace{\underset{j\in I\left(T_{2}\right)}{\sum}\underset{i\in I\left(T_{2}\right)}{\sum}l\left(v_{i},\,v_{j}\right)+\underset{j\in I\left(V\setminus T_{2}\right)}{\sum}\underset{i\in I\left(V\setminus T_{2}\right)}{\sum}l\left(v_{i},\,v_{j}\right)}}.\nonumber
\end{align}

Then by comparing Equality (\ref{eq:T1}) and (\ref{eq:T2}), we obtain
\begin{align}
 & F\left(A\cup\left\{ v_{m}\right\} \right)-F\left(A\right)\label{eq:FA}\\
= & C_{comp}\left(V\setminus T_{2},\text{\,}T_{2}\right)+C_{comm}\left(V\setminus T_{2},\,T_{2}\right)\nonumber \\
= & \left(w_{m}^{edg}-w_{m}^{cld}\right)+\underset{V_{cld}=A\cup\left\{ v_{m}\right\} }{\underbrace{\underset{j\in I\left(A\right)}{\sum}l\left(v_{m},\,v_{j}\right)+\underset{j\in I\left(V\setminus A\setminus T\right)}{\sum}l\left(v_{m},\,v_{j}\right)}}\nonumber \\
 & -\underset{V_{cld}=A}{\underbrace{\underset{j\in I\left(A\right)}{\sum}l\left(v_{m},\,v_{j}\right)-\underset{j\in I\left(V\setminus A\setminus T\right)}{\sum}l\left(v_{m},\,v_{j}\right)}.}\nonumber
\end{align}

\begin{figure*}[!h]
\centering
\subfloat[Input of \emph{Offloading-LCCA} with $T=\left\{ v_{1},\,v_{6}\right\}$
\label{fig:A 7 node graph}]{\centering{}\includegraphics[scale=0.8]{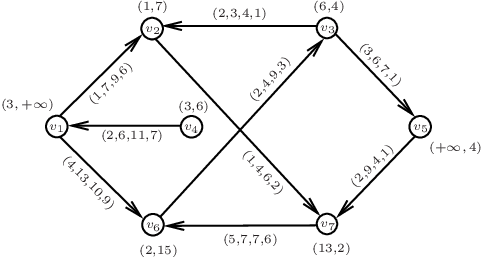}}
$\qquad\qquad$\subfloat[A possible solution for \emph{Offloading-LCCA} \label{fig:a possible offloading solution}]{\includegraphics[scale=0.8]{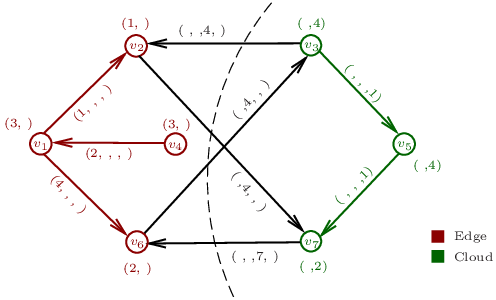}}

\caption{An instance of \emph{Offloading-LCCA} with a solution \label{fig:offloading model}}
\end{figure*}

Similarly, we have
\begin{alignat}{1}
 & F\left(B\cup\left\{ v_{m}\right\} \right)-F\left(B\right)\label{eq:FB}\\
= & \Gamma\left(B\cup\left\{ v_{m}\right\} \right)-\Gamma\left(B\right)\nonumber \\
= & \left(w_{m}^{edg}-w_{m}^{cld}\right)+\underset{V_{cld}=B\cup\left\{ v_{m}\right\} }{\underbrace{\underset{j\in I\left(B\right)}{\sum}l\left(v_{m},\,v_{j}\right)+\underset{j\in I\left(V\setminus B\setminus T\right)}{\sum}l\left(v_{m},\,v_{j}\right)}}\nonumber \\
 & -\underset{V_{cld}=B}{\underbrace{\underset{j\in I\left(B\right)}{\sum}l\left(v_{m},\,v_{j}\right)-\underset{j\in I\left(V\setminus B\setminus T\right)}{\sum}l\left(v_{m},\,v_{j}\right)}.}\nonumber
\end{alignat}
Then, we get the following conclusion by immediately combining Equality (\ref{eq:FA}) with (\ref{eq:FB}) and simple calculation:
\begin{alignat}{1}
 & \left(F\left(A\cup\left\{ v_{m}\right\} \right)-F\left(A\right)\right)-\left(F\left(B\cup\left\{ v_{m}\right\} \right)-F\left(B\right)\right)\nonumber \\
= & \underset{j\in I\left(B\setminus A\right)}{\sum}\Biggl(\underset{V_{cld}=A\cup\left\{ v_{m}\right\} }{\underbrace{l\left(v_{m},\,v_{j}\right)}}-\underset{V_{cld}=B\cup\left\{ v_{m}\right\} }{\underbrace{l\left(v_{m},\,v_{j}\right)}}\Biggr)\nonumber \\
  & +\underset{j\in I\left(B\setminus A\right)}{\sum}\Biggl(\underset{V_{cld}=B}{\underbrace{l\left(v_{m},\,v_{j}\right)}}-\underset{V_{cld}=A}{\underbrace{l\left(v_{m},\,v_{j}\right)}}\Biggr)\text{.}\label{eq:final}
\end{alignat}

\begin{figure*}[ht]\centering
\subfloat[An instance of \emph{Offloading-CommA}\label{fig:An offloading model}]{\centering{}\includegraphics[scale=1.1]{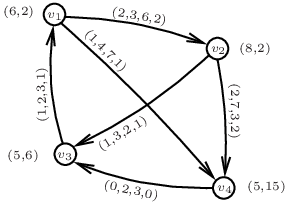}}
\centering{}$\qquad$\subfloat[The output of our algorithm based on submodular minimization\label{fig:submodular}]{\centering{}\includegraphics[scale=1.1]{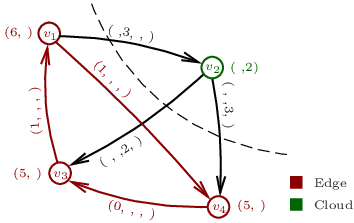}}

$\qquad$\subfloat[\label{fig:min-cut} Failure of  Min-Cut based algorithm:  can not  determine
cost of  blue edges]{\noindent \centering{}\includegraphics[scale=1.1]{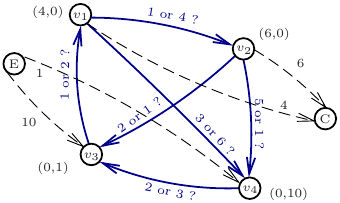}}
\caption{An instance of \emph{Offloading-CommA} with outputs of the algorithms based on submodular minimization and Min-Cut \label{fig:an example min-cut/submodular}}
\end{figure*}

In the above formulation, for the first term of the second line, when
$V_{cld}=A\cup\left\{ v_{m}\right\} $, we have $j\notin I\left(V_{edg}\right)$
since $j\in I\left(B\setminus A\right)$, and hence $m$ and $j$
are in different partitions. So $l\left(v_{m},\,v_{j}\right)$ is inter-partition
communication cost in the case. In contrast, when $V_{cld}=B\cup\left\{ v_{m}\right\} $,
$m$ and $j$ are in the same partition and hence $l\left(v_{m},\,v_{j}\right)$
is intra-cloud communication cost. Following the communication
assumption, for each $j\in I\left(B\setminus A\right)$, we have

\[
\underset{V_{cld}=A\cup\left\{ v_{m}\right\} }{\underbrace{l\left(v_{m},\,v_{j}\right)}}-\underset{V_{cld}=B\cup\left\{ v_{m}\right\} }{\underbrace{l\left(v_{m},\,v_{j}\right)}}\geq0.
\]
So the first term of Equality (\ref{eq:final}) is at least zero, since it sums up the left side of the above inequality over $j\in I\left(B\setminus A\right)$.
Similarly, we also have the second term of Equality (\ref{eq:final})
is not less than zero. This completes the proof.
\end{proof}
By Lemma \ref{lem:min}, \emph{Offloading-LCCA }can be equivalently
solved by computing the following formulation
\begin{equation}
\underset{X\subseteq\left(V\backslash T\right)}{\min}F\left(X\right).\label{eq:submin}
\end{equation}
Moreover, by Theorem \ref{thm:submodular minimization}, the above
formulation is to minimize a submodular function, for which many famous
efficient algorithms exist, including the combinatorial strongly polynomial
algorithm \cite{iwata2001combinatorial}. Hence, we can polynomially
and optimally solve \emph{Offloading-LCCA}.

In particular, when\textit{ $T=\emptyset$}, the \emph{Offloading-LCCA}
problem is actually equivalent to \emph{Offloading-CommA. }That is
because for $T=\emptyset$, two sets $A$ and $B$ for which $A\subseteq B\subseteq V$
holds, and a task $v_{m}\in V\backslash B$, we can similarly get
$F\left(A\cup\left\{ v_{m}\right\} \right)-F\left(A\right)\ge F\left(B\cup\left\{ v_{m}\right\} \right)-F\left(B\right)$.
Thus, for \emph{Offloading-CommA} we have the following corollary:
\begin{cor}
Offloading-CommA can be optimally solved in strongly polynomial time.
\end{cor}

\subsection{Examples of Offloading-LCCA}
\subsubsection{An instance of Offloading-LCCA with a solution}

In Fig. \ref{fig:offloading model}, we demonstrate an instance of the \emph{Offloading-LCCA} problem that consists of
7 nodes with $T=\left\{ v_{1},\,v_{6}\right\}$ together with its
possible solution. In Fig. \ref{fig:A 7 node graph}, the elements of the 2-tuple associated with each node indicate the computation cost of task $v_{i}$ at the edge and the cloud server, respectively. The 4-tuple associated with each directed edge represents the communication cost from task $v_{i}$ to $v_{j}$ of 4 cases: edge to edge, edge to cloud, cloud to the edge, and cloud to cloud, in which the communication assumption $l_{ij}^{1},\,l_{ij}^{4}\le l_{ij}^{2},\,l_{ij}^{3}$
holds. As illustrated in Fig. \ref{fig:a possible offloading solution},
the dashed line indicates a partition dividing the task graph into
two disconnected sets respectively for the edge (in red) and cloud
sever (in green), and each task in $T$ is executed at the edge server.
We use  black lines to denote the communication requirements between the edge and the cloud server, while  the red line for data
communication inside edge servers and the green line for that inside
cloud servers. The  communication cost and computational cost
are labeled as $4$-tuple and $2$-tuple incident to the corresponding
edges and nodes.

\subsubsection{Failure of the Min-Cut Based Algorithm}

We will give a counterexample illustrating how the previous best algorithm \cite{du2020algorithmics}
 fails to solve \emph{Offloading-CommA} when the communication cost is asymmetric. The previous algorithm was shown to work correctly for homogeneous
communication cost (i.e. $l_{ij}^{1}=l_{ij}^{4}\le l_{ij}^{2}=l_{ij}^{3}$), employing a key idea of transforming into Min-Cut.

Fig. \ref{fig:An offloading model} demonstrates an example of 4 tasks
with $l_{ij}^{1}=l_{ij}^{4}\le l_{ij}^{2}\neq l_{ij}^{3}$ for each
link. The computation costs of each node are given in the form of
$\left(w_{i}^{edg},\,w_{i}^{cld}\right)$, where the former and the
latter are the computation costs for executing the task in edge and
cloud server, respectively. Besides, $\left(l_{ij}^{1},\,l_{ij}^{2},\,l_{ij}^{3},\,l_{ij}^{4}\right)$
represent the four costs of the four types of data communication between
task $v_{i}$ and $v_{j}$: edge to edge, edge to cloud, cloud to
edge, and cloud to cloud. Then we execute both our algorithm based
on submodular minimization and the previous algorithm based on Min-Cut.
Fig. \ref{fig:submodular} is the output of our algorithm, which is probably an optimal solution. The dashed line divides the tasks into
two sets, the set of red nodes, say $\left\{ v_{1},\,v_{3},\,v_{4}\right\} $,
will be executed at the edge server and $\left\{ v_{2}\right\} $
the other set (of green nodes) at the cloud server. Consequently,
the red edges turn out to be intra-cloud communication, while the black edges are inter-partition communication. Summing up the computation communication costs, the output of our algorithm consumes a total cost of 28. In contrast, according to the procedure of the algorithm based on Min-Cut \cite{du2020algorithmics}, Fig. \ref{fig:min-cut} is the configuration when the algorithm stops. Thus, the previous
algorithm can not produce a feasible solution for the instance. The reason for the failure is that each edge in the instance could be with four different communication costs, and hence can not transform to the graph in which the Min-Cut algorithm can correctly work. As shown in Fig. \ref{fig:min-cut}, the values of the blue edges can not be determined according to the transformation as in the algorithm based on Min-Cut \cite{du2020algorithmics}.

\begin{figure*}[tbp]
    \centering
    \subfloat[Cost comparison on different datasets]{\includegraphics[width=0.3\textwidth]{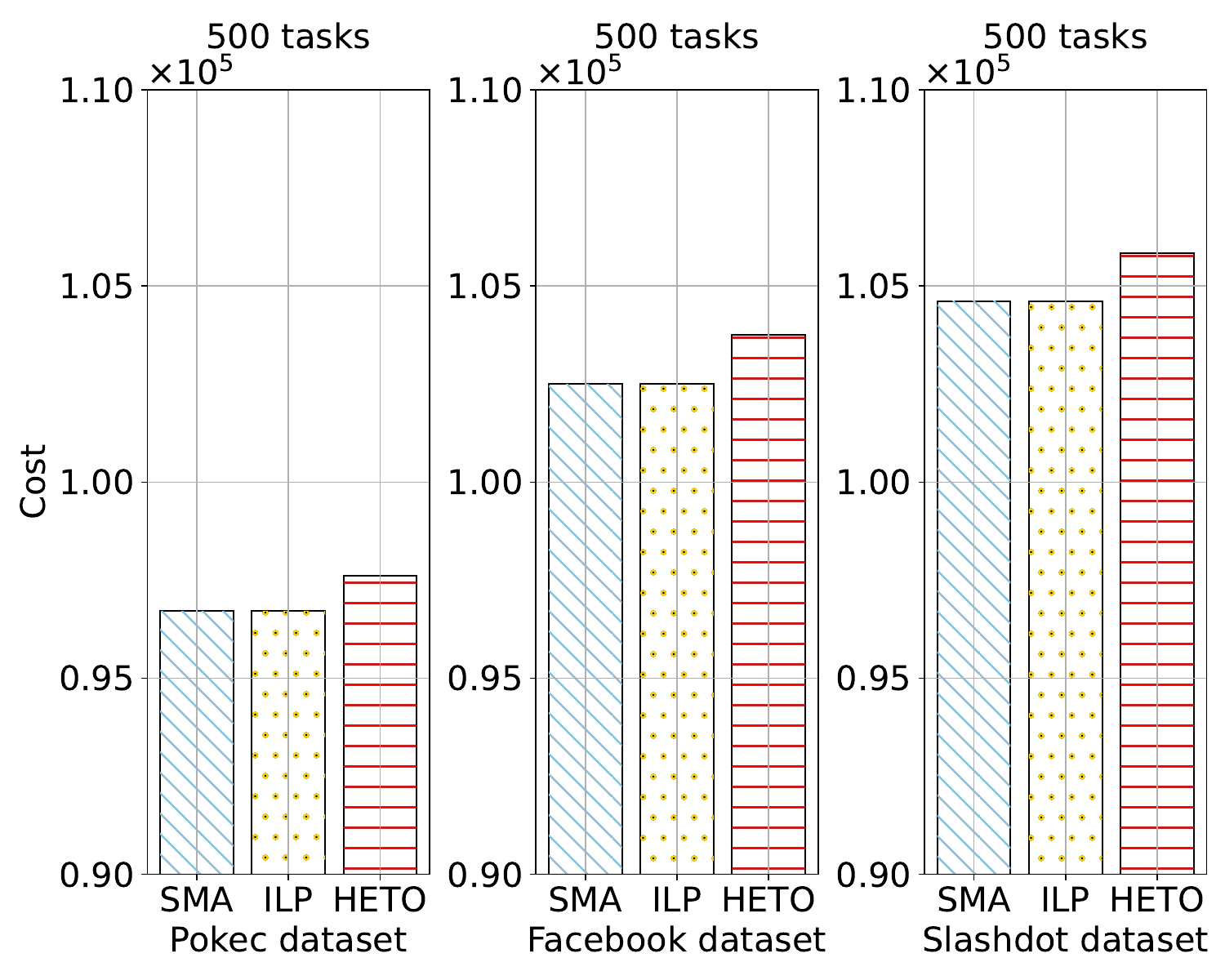}
    \label{fig:snap_different_dataset_cost}} 
    \hfil
    \subfloat[Time efficiency comparison on different datasets]{\includegraphics[width=0.33\textwidth]{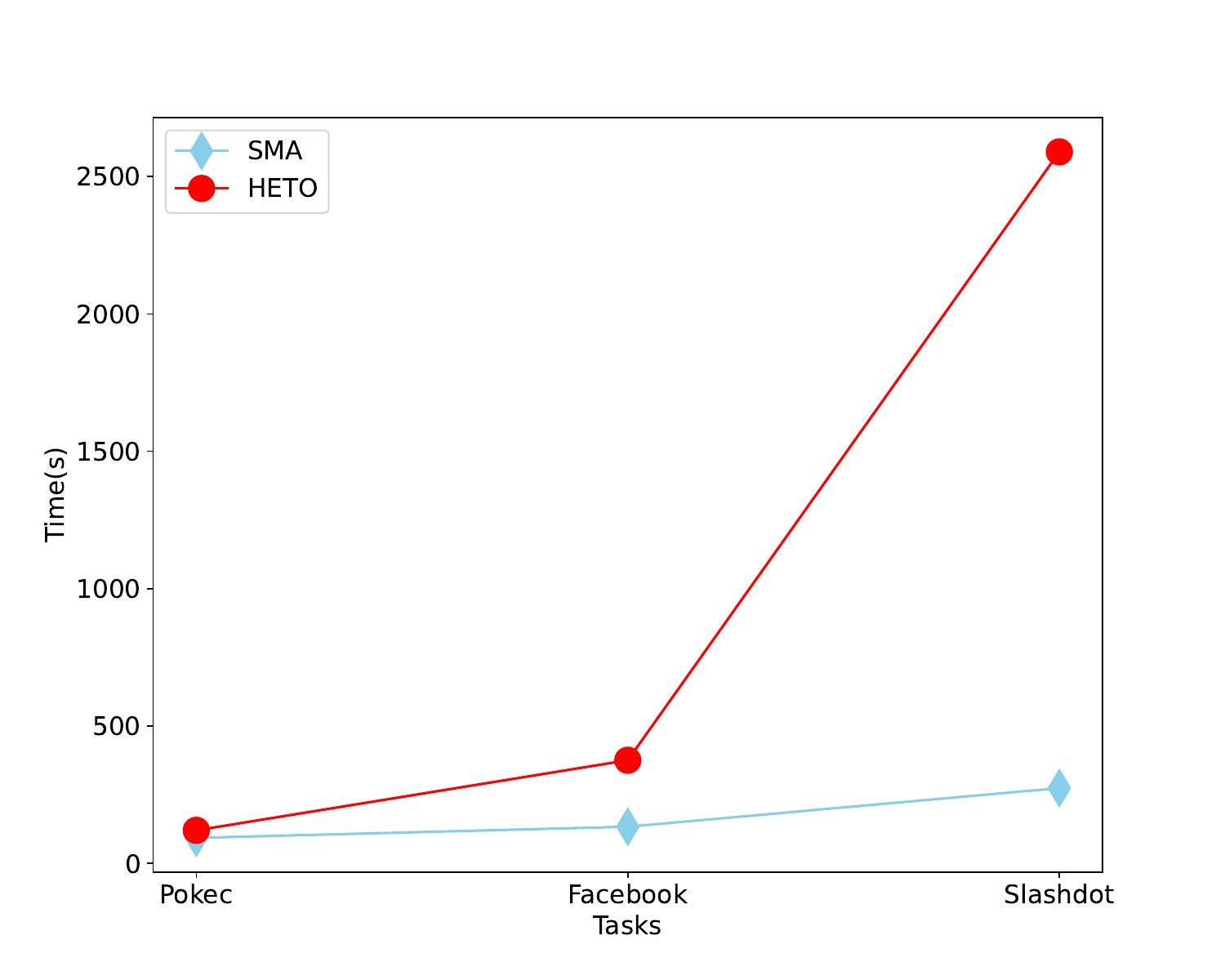}
    \label{fig:snap_different_dataset_time}}
    \caption{Offloading results on the different datasets.}
    \label{fig:snap_different_dataset}
\end{figure*}

\begin{figure*}[tbp]
    \centering
    \subfloat[Cost comparison on different graph size]{\includegraphics[width=0.5\textwidth]{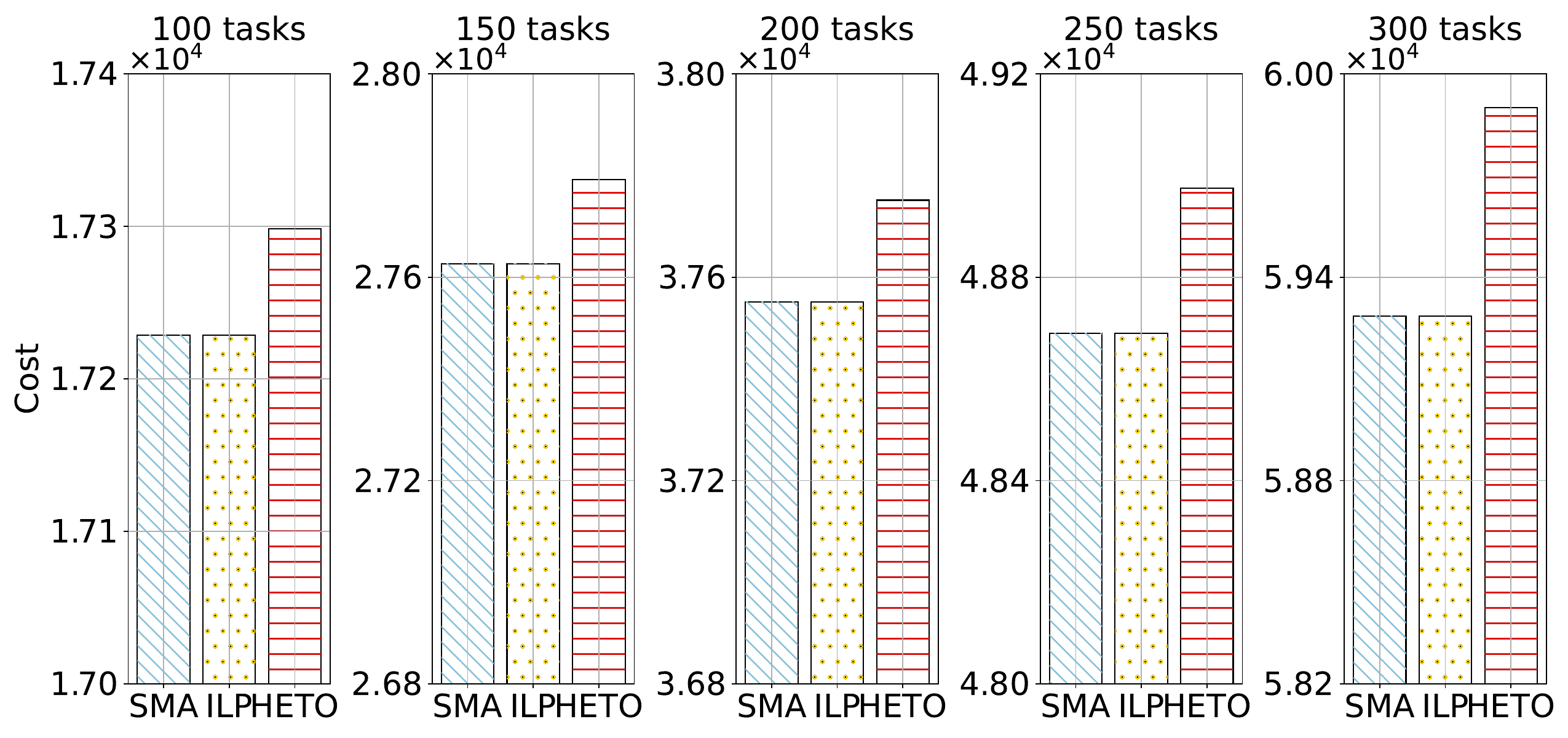}
    \label{fig:snap_nodes_cost}} 
    \hfil \subfloat[Time efficiency comparison on different graph size]{\includegraphics[width=0.36\textwidth]{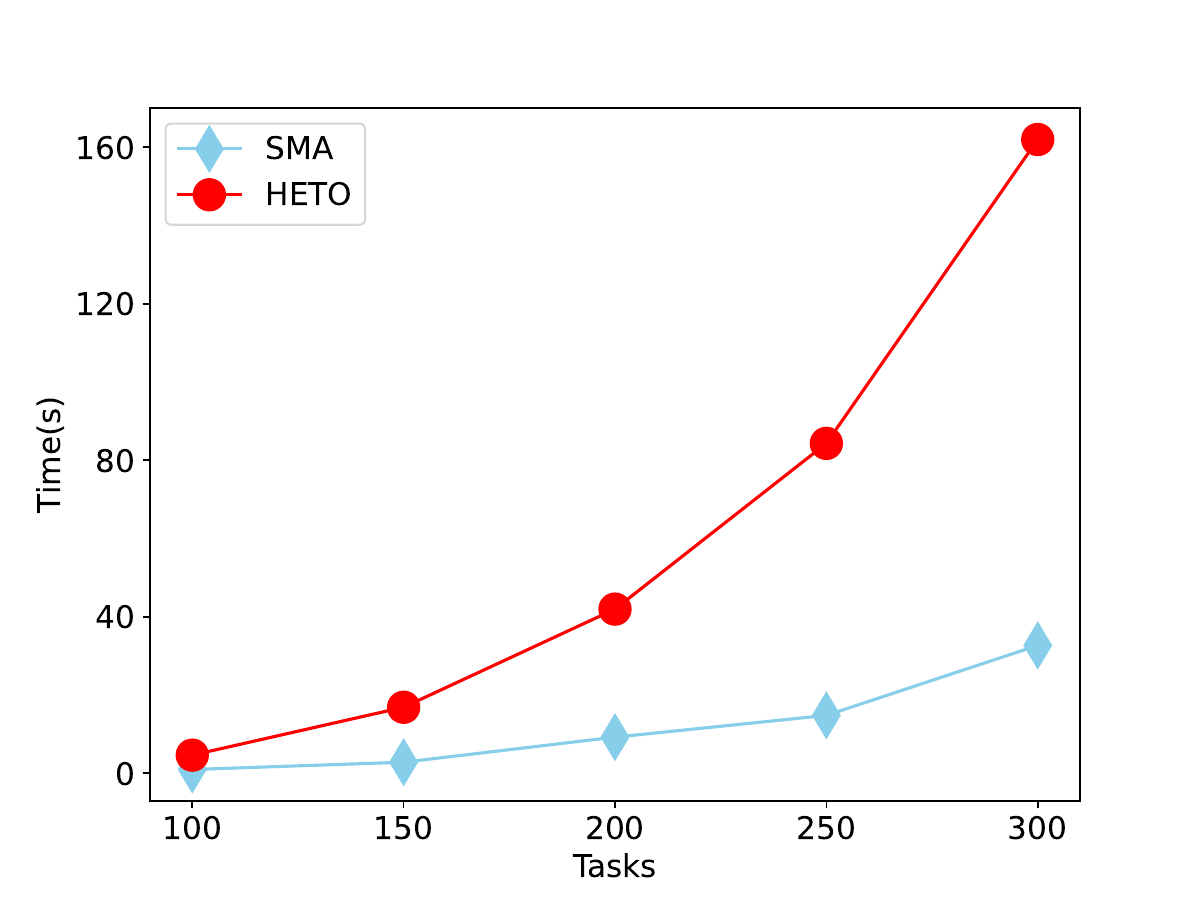}
    \label{fig:snap_nodes_time}}
    \caption{Offloading results over graph datasets of different sizes.}
    \label{fig:snap_cost}
\end{figure*}

\section{Numerical Experiments \label{sec:numerical experiments}}

In this section, we first experiment on several different real-world datasets to evaluate the practical performance of our algorithm based on submodular minimization (denoted by SMA) for \emph{Offloading-CommA}. Then, we further analyze its performance subject to various conditions: (1) Validate the algorithm across different graph sizes to demonstrate its scalability and robustness; (2) Assess the algorithm's performance under different communication cost ratios to verify its adaptability to diverse network conditions; (3) Evaluate the algorithm's efficiency in graphs with varying densities, confirming its capability to handle different levels of complexity. 
All the experiments were conducted under two experimental settings: adhering to the communication assumptions and violating the communication assumptions.

Furthermore,  we validate our SMA’s solution quality by using the solution of Integer Linear Programming (ILP) formulation as a benchmark while also comparing with the HETO algorithm proposed by Du et al. \cite{du2020algorithmics}.  

\subsection{Experimental Setup}
To demonstrate the scalability of the SMA, we conducted experiments using graphs from the SNAP ego-Facebook dataset \cite{leskovec2012learning}, the SNAP soc-Pokec dataset \cite{takac2012data}, the SNAP soc-Slashdot dataset \cite{leskovec2009community} and synthesized datasets. In the experiments, we use graph nodes to represent computing tasks, and edges between nodes to represent communication between tasks. Both computation and communication costs are generated using uniform distributions. 
Each data set was tested across an average of 10 experiments to derive results. All these algorithms were implemented in Python 3.10 while using the SFO toolbox in Matlab \cite{krause2010sfo} to solve the transformed submodular minimization problem optimally. The experiments were carried out on a PC with Microsoft Windows 11, AMD Ryzen 5 4600H, and 16GB 3200MHZ DDR4 Memory.

\subsection{SNAP Dataset}

\descr{Adhering to the Communication Assumption}
\subsubsection{Different dataset}

In this study, we selected three datasets from SNAP: ego-Facebook, soc-Pokec, and soc-Slashdot. From each dataset, we extracted 500 nodes and their connected edges to construct graphs. The number of edges varied across the datasets, with ego-Facebook having 4,337 edges, soc-Pokec having 3,084 edges, and soc-Slashdot having 5,080 edges.

As shown in Fig. \ref{fig:snap_different_dataset_cost}, the costs produced by SMA (depicted in blue) align with those of ILP (depicted in yellow) and are consistently significantly lower than those of HETO (depicted in red). Because it is known that  ILP always provides an optimal solution while SMA and ILP produce identical results, we get that \textbf{SMA can find the optimal offloading strategy when the communication assumptions are valid.} Moreover, the service cost of SMA and ILP is much smaller than that of  HETO. Additionally, Fig. \ref{fig:snap_different_dataset_time} highlights that SMA maintains a lower and more stable runtime across different datasets compared to HETO.

\subsubsection{Different graph sizes}

For this research, we delve into evaluating our algorithm's performance using the SNAP ego-Facebook dataset. Initially, we partitioned the task graphs within the dataset, which contained between 100 and 300 nodes. This approach enabled us to conduct a detailed comparative analysis of task offloading effectiveness among three distinct algorithms: SMA, HETO, and ILP. We specifically focused on understanding how these algorithms perform across various graph sizes, offering insights into their scalability and efficiency in different scenarios.

As illustrated in Fig. \ref{fig:snap_nodes_cost}, the cost incurred by SMA consistently remains significantly lower than that of HETO. As the problem scale increases, the performance gap between the two algorithms widens, clearly indicating the superiority of SMA's offloading strategy over HETO. The Fig. \ref{fig:snap_nodes_time} illustrates the runtime of the SMA, ILP, and HETO across different instances. It can be observed from the figure that SMA requires less time compared to HETO, and the disparity in runtime between the two algorithms increases sharply as the problem size grows. This trend is attributed to the HETO's exponential growth in runtime, which aligns with its time complexity of \( O(E^2) \).  Consequently, \textbf{
SMA not only consistently reaches the optimal solution on different graph sizes, but it also offers a more reduced running time compared to HETO}, facilitating more efficient management of a wide range of datasets.

\begin{figure*}[ht]
    \centering
    \subfloat[Cost comparison for different multiples of communication cost]{\includegraphics[width=0.7526\textwidth]{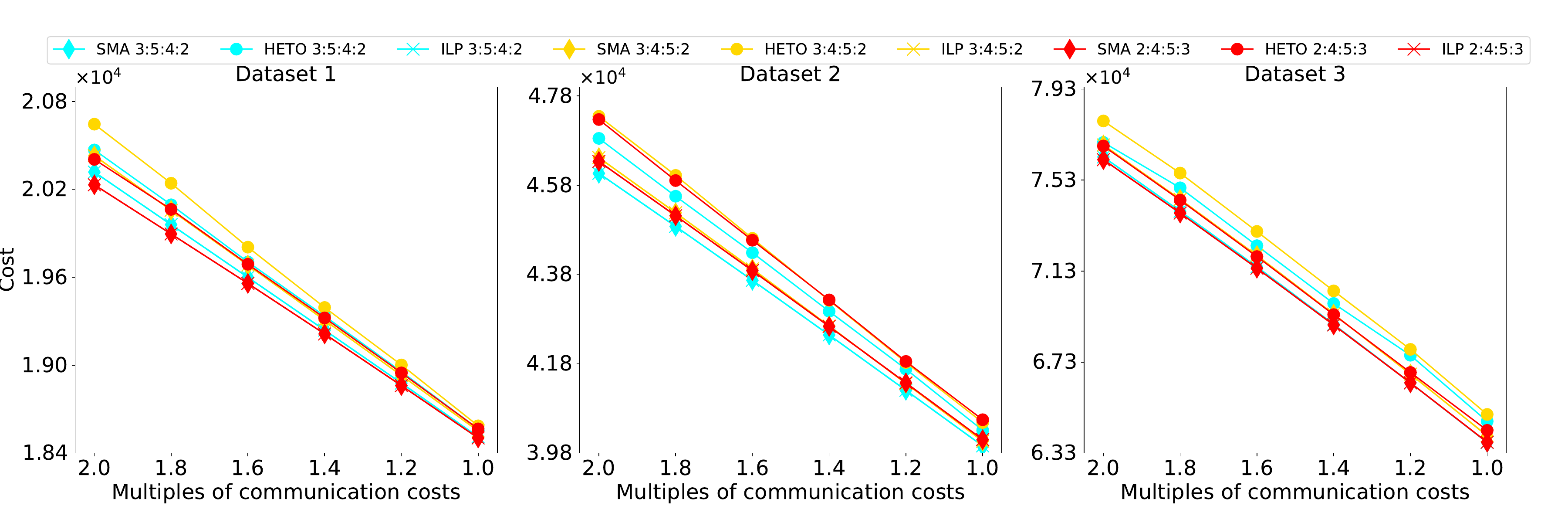}} 
        \hfill
  \subfloat[Time efficiency comparison for different multiples of communication cost]{\includegraphics[width=0.7526\textwidth]{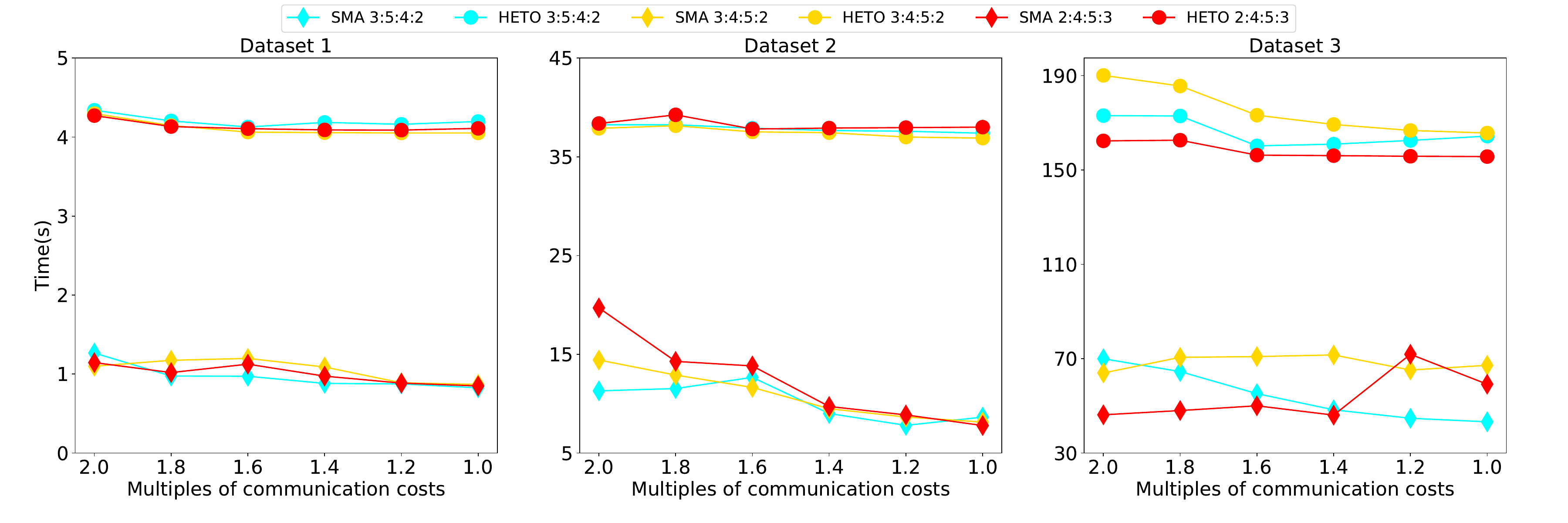}}
        \hfill
    \caption{The influence of communication cost on the performance of algorithms.}
    \label{fig:snap_communication_cost}

\end{figure*}

\subsubsection{Different communication cost ratio}

In this experiment, we will examine how communication costs influence offloading decisions while adhering to communication constraints. To this end, we conducted a series of experiments using three distinct datasets from the ego-Facebook dataset, each consisting of 100, 200, and 300 nodes, respectively. In these experiments, we established varying ratios for communication costs, denoted by \( (l_{ij}^1,l_{ij}^2,l_{ij}^3,l_{ij}^4) \), specifically set at \( (3:5:4:2) \), \( (3:4:5:2) \), and \( (2:4:5:3) \) for each dataset. Moreover, to simulate the enhancement of machine bandwidth in a cloud-edge computing environment, we adjusted the communication cost multipliers for each specified ratio, thereby mimicking the effect of increased bandwidth on offloading efficacy.

As shown in Fig. \ref{fig:snap_communication_cost}, the lines of the same color denote the same ratio of communication costs. When the ratio of intra-cloud and inter-partition communication costs remains fixed and scaled by a constant factor, SMA consistently yields the same offloading strategy with the lowest total cost as ILP, compared to HETO. Additionally, different colors denote varying cost ratios, \textbf{the consistent superiority of SMA's performance, regardless of fluctuations in communication cost ratios.} It underscores its efficacy as a solution for optimizing task offloading amidst diverse and evolving communication constraints in cloud-edge computing environments.
\vspace{\topsep}\\
\descr{Violating the Communication Assumption}

Although our communication assumptions are designed to reflect real-world application scenarios, we have still conducted rigorous testing of the model under conditions where these assumptions may not hold. To this end, we employed the SNAP ego-Facebook graph datasets, featuring 100, 200, and 300 nodes, to conduct a series of experiments with varied communication cost ratios: \((8:6:7:5)\), \((8:5:6:7)\), and \((7:6:5:8)\). The experiments with the ratios \((8:5:6:7)\) and \((7:6:5:8)\) were intentionally chosen to challenge the model under conditions where the communication assumption is violated, thereby evaluating the model's adaptability and robustness in complex environments.

The experimental outcomes in Fig. \ref{fig:snap_ratio} reveal that SMA achieves optimal results when the communication costs satisfy the communication assumption, as seen with \((8:6:7:5)\), consistent with prior theoretical findings. Conversely, under scenarios \((8:5:6:7)\) and \((7:6:5:8)\) where the communication assumption is not true, SMA does not guarantee to compute the optimal solution, provided that the problem is NP-hard at the point. Nevertheless, the performance of SMA still surpasses HETO. Additionally, for each communication cost ratio, SMA consistently outperforms HETO in terms of cost across different graph sizes. Moreover, Fig. \ref{fig:snap_ratio_time_new} demonstrates that \textbf{even without the communication assumption, SMA compares favorably in time efficiency, computing near-optimal solutions faster than HETO.} This highlights that SMA outperforms the state-of-the-art HETO in efficacy and resilience for negative scenarios in which the communication assumption is not true. 


{
\begin{figure*}[h]
    \centering
    \subfloat[Cost comparison on 100 tasks]{\includegraphics[width=0.326\textwidth]{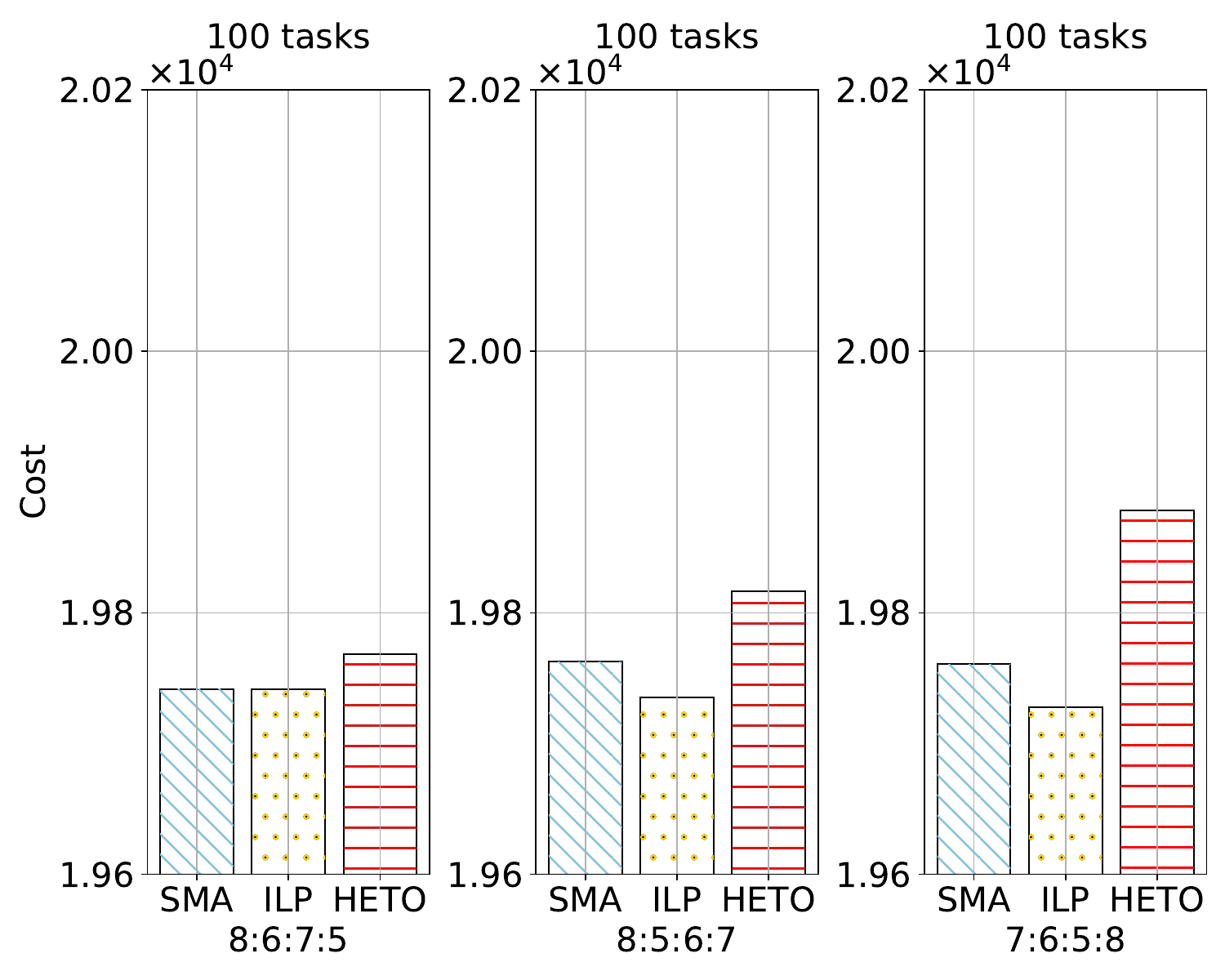}} 
    \hfil
    \subfloat[Cost comparison on 200 tasks]{\includegraphics[width=0.326\textwidth]{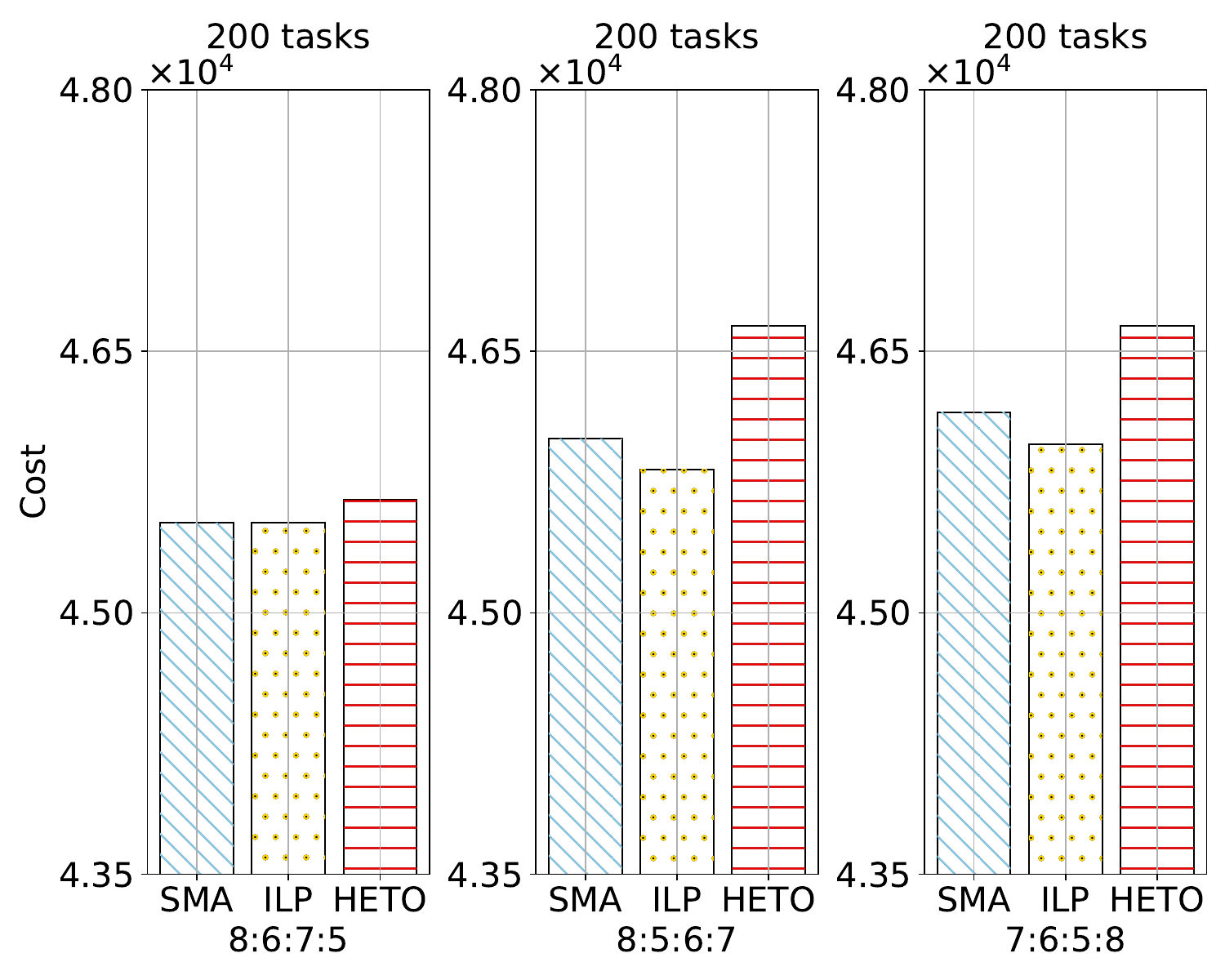}}

   \subfloat[Cost comparison on 300 tasks]{\includegraphics[width=0.326\textwidth]{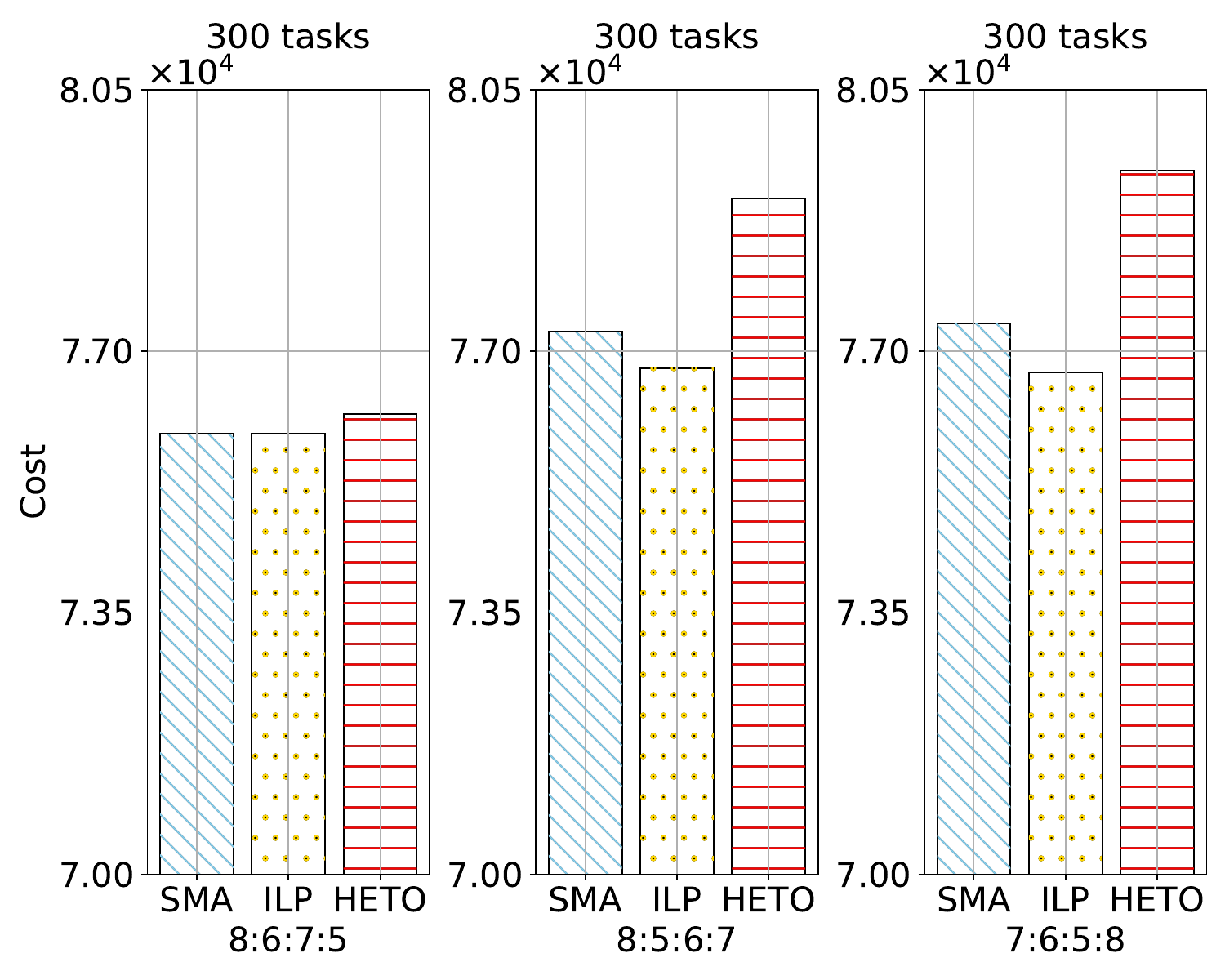}} \hfil
     \subfloat[Time efficiency comparison on different task counts]{\includegraphics[width=0.326\textwidth]{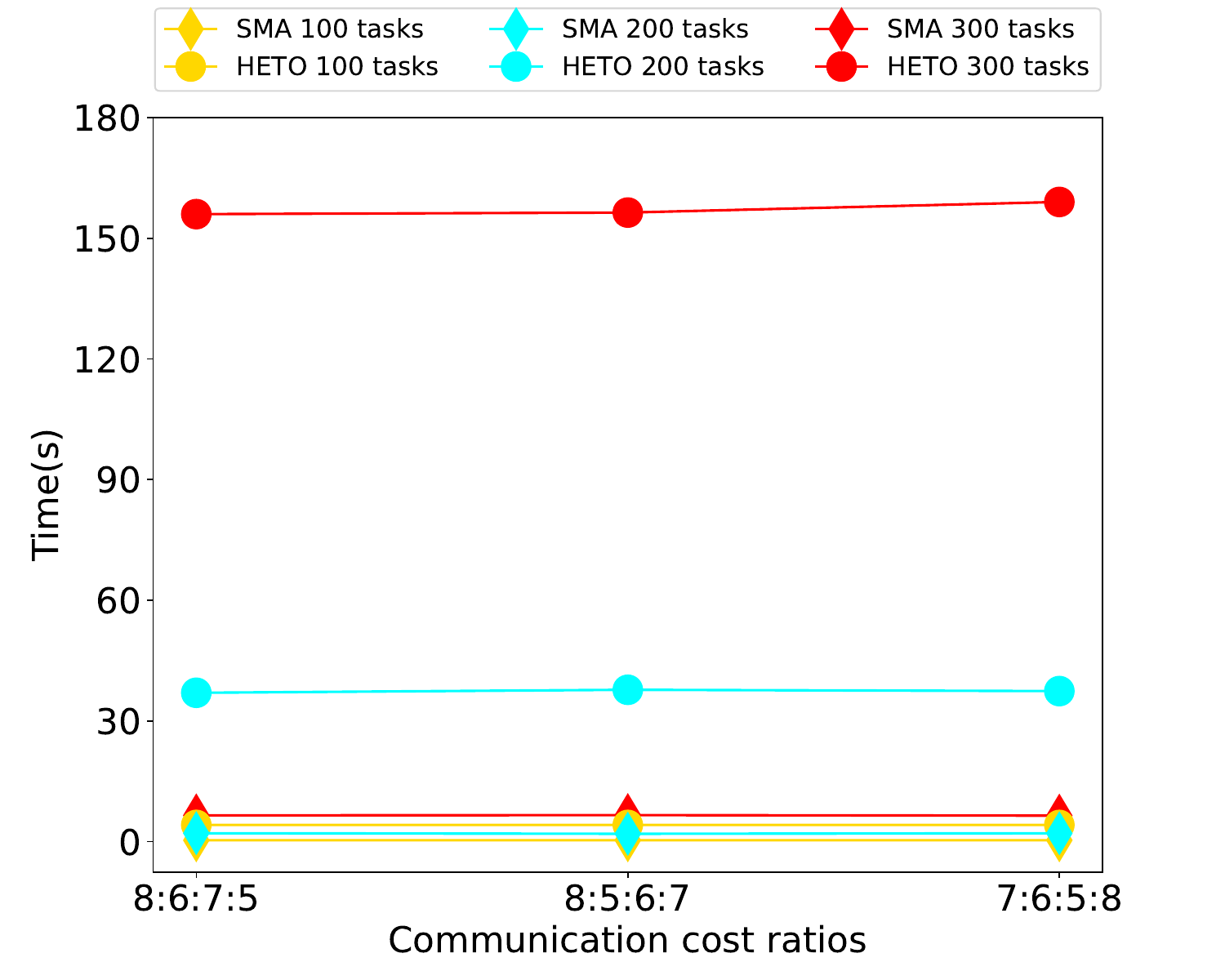}
     \label{fig:snap_ratio_time_new}}
    \caption{Comparison of offloading results without communication assumption.}
    \label{fig:snap_ratio}
\end{figure*}
\begin{figure*}[!h]
    \centering
    \subfloat[Adhering to the communication assumption]{\includegraphics[width=0.7\textwidth]{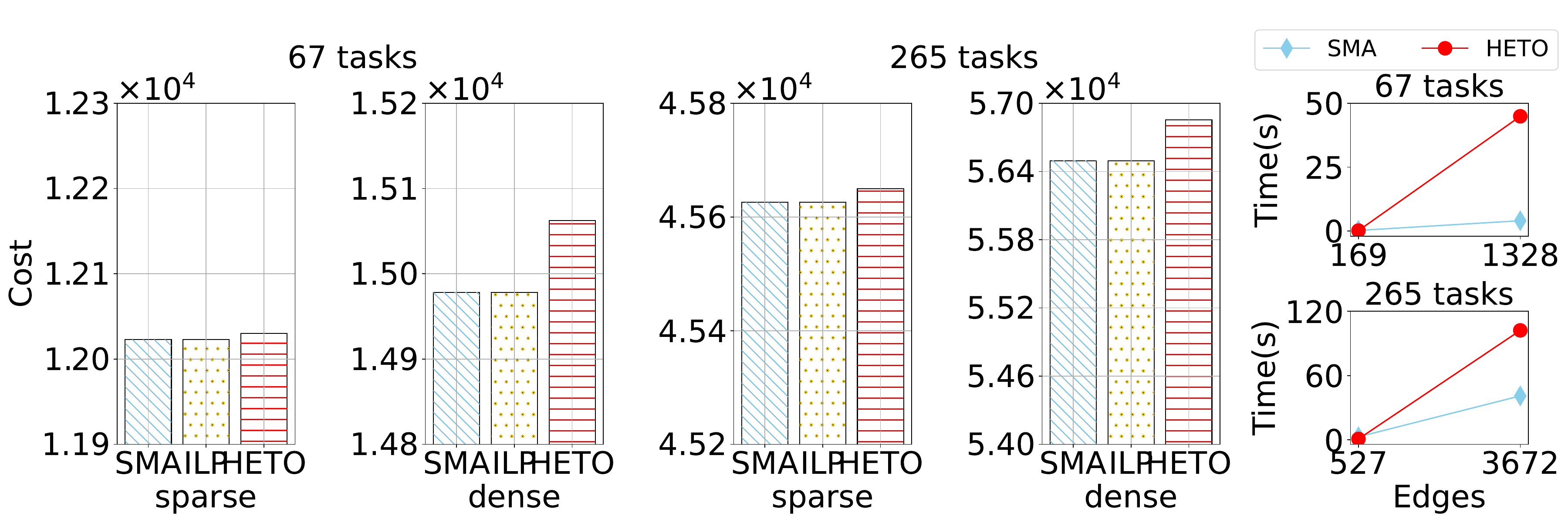}
    \label{figh::syn_density}} 
    \\
    \subfloat[Violating the communication assumption]{\includegraphics[width=0.7\textwidth]{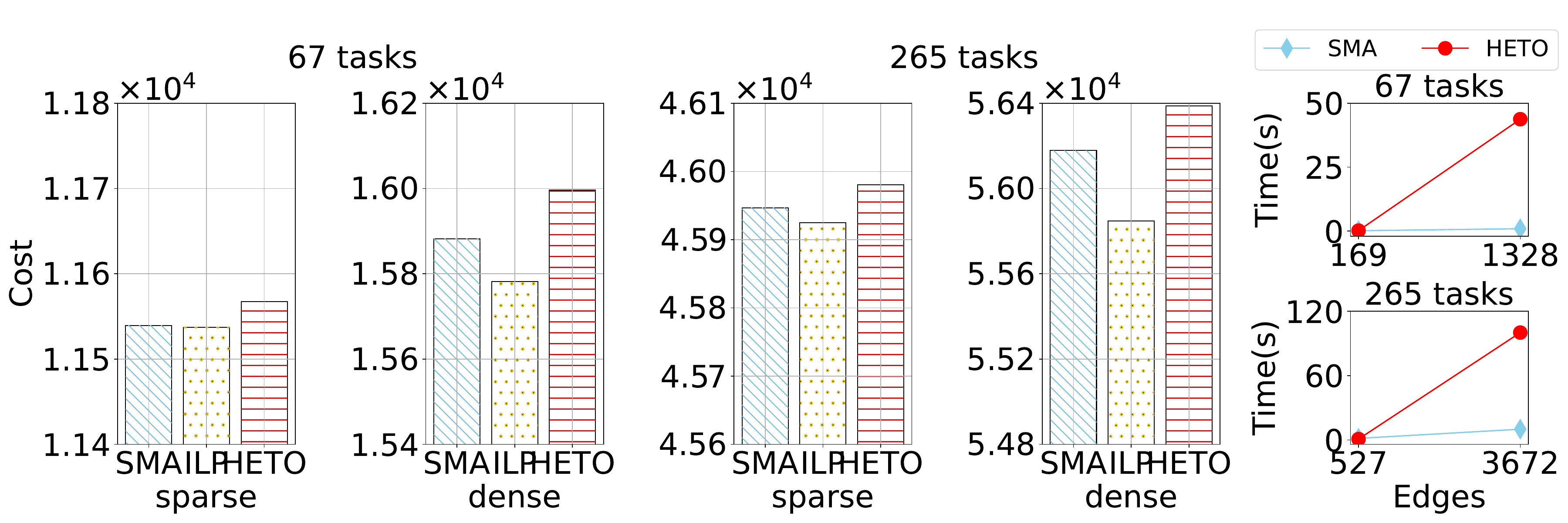}
    \label{figh::syn_density_violate}}
    \caption{Offloading results comparison over the graph of different densities.}
    \label{fig:syn_cost}
\end{figure*}
}

\subsection{Synthesized Datasets}
In this part, we evaluate our algorithm across graphs of varying densities, a vital aspect considering the pivotal role of node-to-node communication in computational offloading as emphasized in \cite{du2020algorithmics}. To facilitate this analysis, we devised a synthetic dataset designed to mirror variations in inter-node communication by modulating the number of edges, which in turn, alters the graph's density. We employed a uniform distribution approach for assigning weights to the nodes and edges to maintain consistency in our evaluation.

Our analysis focuses on two distinct graph pairs to demonstrate the algorithm's performance under different density scenarios. The initial graph pair consists of 67 nodes, where the sparse version contains 169 edges, escalating to 1,328 edges in the dense graph.  The subsequent pair features a graph with 265 nodes, starting with 527 edges in its sparse form and expanding to 3,672 edges in the dense version. This allows us to analyze how our algorithm adapts and performs as the graph density varies.

\vspace{\topsep}
\descr{Adhering to the Communication Assumption}

Fig. \ref{figh::syn_density}. illustrates the outcomes of experiments under the communication assumptions. It demonstrated that SMA outperformed HETO in offloading cost and running time for both sparse and dense graphs. Since the time complexity of the HETO is $O(E^2)$, its running time is affected by the number of edges, so SMA is more suitable for dense graph scenarios.

\vspace{\topsep}
\descr{Violating the Communication Assumption}

In  Fig. \ref{figh::syn_density_violate}, we depict the experimental results where the communication assumption is violated. We find that the performance of HETO falls at around $60\%$, whereas  SMA and ILP both attain optimum. 
Hence, it is evident that \textbf{SMA's offloading approach continues to surpass HETO in terms of performance, regardless of whether the communication assumption is compromised, applicable to sparse and dense graphs.}

As shown in the above experiments, our new model exhibits significant performance gain compared to the current state-of-the-art algorithm, demonstrating its effectiveness in practical applications.

\section{Conclusion \label{sec:Conclusion}}
In this paper, we propose a communication assumption for the offloading problem in the edge-cloud environment, tailored to real-world applications,  with the goal of minimizing the total service cost, which consists of both computation and communication costs. We first show that solving the offloading problem becomes NP-hard—even with symmetric edge costs—when the communication assumption does not hold, through a reduction from MAX-CUT. As our main result, we show that with communication assumption in place, the offloading problem can be transformed to submodular minimization, and is thus polynomially solvable. Notably, this polynomial solvability persists even under additional real-time constraints. By combining these findings regarding polynomial solvability and NP-hardness, we revealed that the communication assumption is the crucial factor that determines the difficulty of the offloading problem. Lastly, we evaluated the practical performance of our algorithm through extensive experiments, demonstrating that it significantly outperforms the state of the art in practice.

\section*{Acknowledgment}

This work is supported by the National Natural Science Foundation of China (Nos. 12271098 and 61772005).


%

\bibliographystyle{IEEEtran}
\bibliography{offloadingref}

\end{document}